\documentclass[aps,prl,groupedaddress,showpacs,twocolumn,nofootinbib,superscriptaddress,longbibliography]{revtex4-1}
\usepackage{times}
\usepackage{amsmath}
\usepackage{amssymb,bm}
\usepackage{amsthm,color,xcolor,dsfont}
\usepackage[colorlinks=true,linkcolor=blue,urlcolor=blue,citecolor=blue,hyperindex,breaklinks]{hyperref}
\usepackage{ulem}
\newtheorem{theorem}{Theorem}

\newtheorem{definition}{Definition}
\newtheorem{proposition}[theorem]{Proposition}

\newtheorem{remark}{Remark}

\newtheorem*{definition*}{Definition}

\begin{document}

\title{Nonzero Classical Discord}

\author{Vlad Gheorghiu}
\email{vgheorgh@gmail.com}
\affiliation{Institute for Quantum Computing, University of Waterloo, 
Waterloo, ON, N2L 3G1, Canada}
\affiliation{Institute for Quantum Science and Technology,
University of Calgary, Calgary, AB, T2N 1N4, Canada}

\author{Marcos C. de Oliveira}
\email{marcos@ifi.unicamp.br}
\affiliation{Instituto de F\'{\i}sica Gleb Wataghin, Universidade Estadual de Campinas, CEP 13083-859,
Campinas, SP, Brazil}
\affiliation{Institute for Quantum Science and Technology,
University of Calgary, Calgary, AB, T2N 1N4, Canada}

\author{Barry C. Sanders}
\email{sandersb@ucalgary.ca}
\affiliation{Institute for Quantum Science and Technology,
University of Calgary, Calgary, AB, T2N 1N4, Canada}
\affiliation{Institute for Quantum Computing, University of Waterloo, 
Waterloo, ON, N2L 3G1, Canada}
\affiliation{Program in Quantum Information Science, Canadian Institute for Advanced Research, Toronto, Ontario M5G 1Z8, Canada}

\date{Version of \today}

\begin{abstract}
Quantum discord is the quantitative difference
between two alternative expressions for bipartite mutual information, given respectively in terms of two distinct definitions for the conditional entropy.
By constructing a stochastic model of shared states,
classical discord can be similarly defined, quantifying the presence of some stochasticity in the measurement process. Therefore, discord can generally be understood as a quantification of the system's state disturbance due to local measurements, be it quantum or classical.
We establish an operational meaning of classical discord in the context of state merging
with noisy measurement
and thereby show the quantum-classical separation 
in terms of a negative conditional entropy.

\end{abstract}

\maketitle
Entanglement exemplifies the mystery of quantum mechanics,
for example the conundrum of Schr\"{o}dinger's cat~\cite{Sch35a,Sch35b,Sch35c},
and embodies the quintessential resource for quantum information processing,
such as the consumable ebits for quantum teleportation~\cite{BBC+93}.
Recently significant effort is expended on extending the notion of entanglement
to generalized bipartite ``quantum correlations''~\cite{RevModPhys.84.1655}
with ``quantum discord'' the most ubiquitous of these measures
and furthermore universal in the sense that only a negligibly few states
have zero discord~\cite{PhysRevA.81.052318}.

The huge effort into studying discord is driven by the optimism that,
``Algorithms could instead tap into a quantum resource called discord, which would be far cheaper and easier to maintain in the lab''~\cite{Mer11}.
Just as entanglement is operationalized~\cite{Bri27}
by teleportation,
discord is operationalized by state merging~\cite{HOW05,HOW07,PhysRevA.83.032324,MS11}.
Therefore, discord can be understood not just as a mathematical characterization of a state
but also as a quantity relevant for performing a certain information task.

The question we address is whether this resource,
operationalized by state merging,
is restricted only to the quantum world or can be described within a classical theory.

To assess the quantumness of discord,
we need to transcend the dichotomy of quantum vs.\ classical mechanics,
in which something that is not classical,
in the sense of deterministically evolving objects with arbitrarily precise properties,
is \emph{ipso facto} quantum. Similar issues have been recently raised in the context
of weak measurements \cite{PhysRevLett.113.120404}.
Development of a quantum ``toy theory''~\cite{Spe07} is an example of breaking this dichotomy by finding a self-consistent theory that is broader than classical theory yet not as powerful as quantum theory.
A famous middle-ground approach arises by replacing classical mechanics by stochastic mechanics,
or stochastic field theory~\cite{Mil76},
in which case either or both the precision of specifying states and the dynamics are sacrificed.
Those theories do not rule out quantum phenomena such as quantum coherence and entanglement; instead they give an alternative description or possibility of emergence of analogous behavior in terms of purely classical stochastic processes.

The essential feature being captured by the quantum discord is how much a bipartite system state is affected by measurements. The peculiar structure of quantum measurements intrinsically leads to conditioned state changes when a detection is made and naturally introduces the notion of correlations that cannot be locally accessed even when there is no entanglement. There is a clear distinction of this situation from a classical one, when perfect measurements are enacted. However a scenario of imperfect classical measurements can be conceived  where such an analogous situation to the quantum one could exist. The purpose of this Letter is to propose such a scenario by introducing the notion of classical discord, which shares the same feature of its quantum counterpart, in signaling how much a system state is disturbed by measurements.

We begin by reviewing the essence of discord.
Consider two parties sharing some state of information.
A precisely known state of information can be represented by an $n$-bit string
$\bm{i}\in\{0,1\}^n$,
which we now generalize to a state in stochastic information theory.
\begin{definition}
A stochastic-information state is
a distribution $p(\bm{i})\geq 0$ such that $\sum p(\bm{i})=1$.
\end{definition}
The Shannon entropy of this stochastic-information state is $H=-\sum p(\bm{i})\log p(\bm{i})$
with the base~2 logarithm notation suppressed.
In the quantum case the state is a trace-class bounded positive operator~$\rho$
on the Hilbert space represented by $\mathbb{C}_2^{\otimes n}$ with basis
$\left\{|\bm{i}\rangle\right\}$,
and the state is pure if $\rho^2=\rho$.
The state's von Neumann entropy is $S=-\text{tr}(\rho\log \rho)$.

Mutual information relates to shared information between two parties often called ``Alice'' (A)
and ``Bob'' (B).
In the stochastic classical case the joint state can be described by a probability matrix
\begin{equation}
\label{eq:pAB}
	p_{AB}
		=\sum_{\bm{i},\bm{j}}p_{AB}\left(\bm{i},\bm{j}\right)
			\chi(\bm{i},\bm{j}),
\end{equation}
where $\bm{i}$ labels the rows and $\bm{j}$ labels the columns, and the $\bm{i},\bm{j}$ entry $p_{AB}(\bm{i},\bm{j})$ represents the joint probability of Alice having the bit string $\bm{i}$ and Bob having the bit string $\bm{j}$, respectively.
Here~$\chi(\bm{i},\bm{j})$ represents the $2^n\times 2^n$ matrix with zeros everywhere except at the location $(\bm{i}, \bm{j})$,
which is the classical version of the quantum basis tensor product
\begin{equation}
\label{eq:chi}
	\hat{\chi}(\bm{i},\bm{j})
		:=|\bm{i}\rangle\langle\bm{i}|\otimes|\bm{j}\rangle\langle\bm{j}|.
\end{equation}
The strings~$\bm{i}$ and~$\bm{j}$ need not be equal length, but we make the restriction
that Alice and Bob each hold $n$-bit strings for convenience of notation.
In the quantum case the joint state~$\rho_{AB}$ acts on $\mathbb{C}_2^{\otimes n}\otimes\mathbb{C}_2^{\otimes n}$. 
The notation we used in \eqref{eq:pAB} can be generalized to encompass tripartite or multipartite stochastic classical states. For example, in the tripartite case, such a state can be described by a rank-3 probability tensor $p_{ABC}=\sum_{\bm{i},\bm{j},\bm{k}}p_{ABC}(\bm{i},\bm{j},\bm{k})\chi(\bm{i},\bm{j},\bm{k})$, where $\chi(\bm{i},\bm{j},\bm{k})$ is now a rank-3 tensor, representing the classical equivalent of the quantum basis tensor product $|\bm{i}\rangle\langle\bm{i}|\otimes|\bm{j}\rangle\langle\bm{j}|\otimes|\bm{k}\rangle\langle\bm{k}|$.

Let~$H(A)$ signifying the entropy of~A's marginal distribution and similar for~B and for A and~B.
The mutual information $I(A;B)=H(A)+H(B)-H(A,B)$
between Alice and Bob is mathematically equivalent to $J(A;B):=H(A)-H(A|B)$
for
\begin{equation}
\label{eq:HAB}
	H(A|B)
		=-\sum_{\bm{i}}p_{A}(\bm{i})
			\sum_{\bm{j}}p_{B}(\bm{j}|\bm{i})\log p_{B}(\bm{j}|\bm{i}).
\end{equation}
Operationally,
this conditional entropy is obtained by Bob measuring and announcing his results to Alice.

If Alice's state has zero entropy conditioned on Bob's results,
and vice versa for Bob's state conditioned on Alice's results,
we refer to this state as conditionally pure.
\begin{definition}
\label{def:pure}
The stochastic-information joint state~$p_{AB}\left(\bm{i},\bm{j}\right)$
is conditionally pure iff $H(A|B)=0$.
\end{definition}
\begin{remark}
The state is trivially pure if the Shannon entropy of the stochastic-information state is zero,
but this notion of conditional purity allows noisy preparation provided that
Alice and Bob can still recover pure states by one-way communication.
The state would not be conditionally pure if either Alice or Bob's announcements are based
on results from noisy measurements.
\end{remark}
Note that the conditional entropy is zero if and only if the results of the measurements on $B$ are completely determined by $A$, or, equivalently, if the marginal distribution on $B$ is a function only of $A$ (see Ch.~2 of \cite{CoverThomas:ElementsOfInformationTheory}). Hence, according to our definition, a stochastic-information state is conditionally pure if and only if it has the form
\begin{equation}\label{eqnv2}
p_{AB} = \sum_{\bm{i}}p_{AB}\left(\bm{i},f(\bm{i})\right) \chi(\bm{i},f(\bm{i})),
\end{equation}
where $f(\bullet)$ can be any bijection on bit strings of length $n$. A simple two bit example of a conditionally pure state is $p_{AB}=(1-p)\chi(0,1)+p\chi(1,0)$, with $0\leqslant p \leqslant 1$.

In the quantum case,
the same expression for~$I$ above holds
but with the Shannon entropy~$H$ being replaced by the von Neumann entropy~$S$. In that situation the expression for the conditional entropy~\cite{NC00}
$S(A|B):=S(A,B)-S(B)$, can be negative, signaling the presence of entanglement, 
thereby making its operational interpretation less straightforward.

To address the operational challenge of quantum conditional entropy,
an alternative version explicitly dependent on one party's choice of measurements
has been introduced~\cite{OZ01}.
For $\Pi:=\{\pi_k\}$ representing Bob's choice of projective measurement basis
and
\begin{equation}
	\rho_{A|k}
		:=\frac{\text{tr}_{B}[(\mathds{1}\otimes \pi_k)\rho_{AB}]}
			{\text{tr}[(\mathds{1}\otimes \pi_k)\rho_{AB}]}
\end{equation}
being the resultant states conditioned on Alice's side,
this measurement-dependent conditional quantum entropy is
\begin{equation}
	S^\Pi(A|B):=\sum_{k}p_kS(\rho_{A|k})\geqslant 0,
\end{equation}
where $p_k=\text{tr}[(\mathds{1}\otimes \pi_k)\rho_{AB}]$ denotes the probability of Bob obtaining the result $k$ when measuring the state $\rho$ in the basis $\{\pi_k\}$.
Operationally this nonnegative quantity represents how much information,
on average,
 Alice can extract from her system given Bob's measurement results.
 
Quantum discord, defined as~\cite{OZ01}
\begin{equation}
	\mathcal{D}_{A\rightarrow B}(\rho_{AB}):=\min_\Pi S^\Pi(A|B) - S(A|B) \label{disc}
\end{equation}
is equivalent to
\begin{equation}
\label{eq:DIJ}
	\mathcal{D}_{A\rightarrow B}(\rho_{AB})=I(A;B) - \max_\Pi J^\Pi(A;B),
\end{equation}
with
\begin{equation}
	J^\Pi(A;B):=S(A) - S^\Pi(A|B),
\end{equation}
being the version of mutual information
that employs measurement-dependent quantum conditional entropy.
The term~$J^\Pi(A;B)$ is the ``classical'' part of the correlations;
hence, nonzero quantum discord is attributed to genuine quantum correlations 
independent of measurements regardless of the measurement basis.
In this perspective,
zero discord is proclaimed as being classical due to $I\equiv J$. 
An interpretation of quantum discord as given by Eq.~\eqref{disc} is that of a measure of how much the bipartite quantum system is affected by local measurements. We argue in the following that such interpretation is directly applicable to the classical case, where we define the \emph{classical discord}, and show that its properties closely resembles the ones of quantum discord

Our aim is to show that nonzero discord can hold classically by treating
conditional entropy operationally and by incorporating noise into Bob's measurement process,
recognizing that any measurement procedure cannot be perfect.

In other words we now establish that $I\not\equiv J$ if measurement is at all noisy.
This noisy measurement arises naturally within
the stochastic reconciliation protocol used
to construct the conditional entropy.

In contrast with the definition of quantum discord, which invokes a minimization over all possible bases \eqref{disc}, in classical information theory there is only one measurement basis. Considering that in reality channels are never ideal, and are in fact characterized by parameters that specify how ``noisy" they are, such as channel capacity or rate distortion~\cite{CoverThomas:ElementsOfInformationTheory}, we treat Bob's measurement as always noisy.
This is equivalent to Bob using a noisy channel~$\mathcal M$,
represented by the stochastic transition matrix~$M$,
followed by perfect measurements of his bit string. Unlike the quantum case,
with uncountable basis choices~$\{\Pi\}$,
the classic noiseless measurement can be performed only in one basis~$\{\bm{i}\}$.
As a stochastic matrix,
$\sum_{\bm{i}}M(\bm{i},\bm{j})=1~\forall\bm{j}$.
The noiseless and maximally noisy cases correspond to $M=\mathds{1}$
and to $M(\bm{i},\bm{j})=2^{-n}~\forall\bm{i},\bm{j}$, respectively.

Now we define classical discord analogous to the quantum discord~\eqref{eq:DIJ}
but with quantum measurement-based mutual information~$J^\Pi$
replaced by a noisy classical measurement-based counterpart defined below.
The noisy measurement apparatus is represented by a channel~$\mathcal{M}$,
which can be represented by a stochastic matrix~$M$.

\begin{definition}
The stochastic-information state with added stochasticity
due to Bob's measurement apparatus
(quantified by stochastic matrix~$M$)
is
\begin{equation}
	p_{AB'}:=p_{AB}M^\mathrm{T},
\end{equation}
for~$\mathrm{T}$ denoting the matrix transpose,
and the subscript~$B'$ denotes that the noise is on Bob's side.
\end{definition}
\begin{definition}
The mutual information of noisy state~$p_{AB'}$ is
\begin{equation}
\label{eq:JM}
	J^\mathcal{M}(p_{AB}):=I\left(p_{AB'}\right)=I\left(p_{AB}M^\mathrm{T}\right).
\end{equation}
\end{definition}
\begin{definition}
The classical discord of state~$p_{AB}$
subjected to B's measurement
with noise~$\mathcal{M}$ represented by stochastic matrix~$M$,
is
\begin{equation}
\label{eq:DMABpAB}
	\mathcal{D}_{A\rightarrow B}^{\mathcal{M}}(p_{AB}):=I(A;B) - J^\mathcal{M}(p_{AB}).
\end{equation}
\end{definition}

\begin{proposition}
Classical discord is nonnegative.
\end{proposition}
\begin{proof}
$\mathcal{D}^\mathcal{M}_{{A}\to{B}}(p_{AB})\geqslant 0,~\forall p_{AB}$
follows immediately from the data processing inequality~\cite{CoverThomas:ElementsOfInformationTheory}.
\end{proof}

As any classical~\eqref{eq:pAB} can be ``embedded" into a quantum bipartite state~$\rho_{AB}$
by replacing~$\chi(\bm{i},\bm{j})$ by the quantum projector~\eqref{eq:chi},
namely
$\rho_{AB}=\sum_{\bm{i},\bm{j}}p_{AB}(\bm{i},\bm{j})
\hat\chi(\bm{i},\bm{j})$,
the necessary and sufficient conditions for zero quantum discord \cite{OZ01}
must also be necessary and sufficient for classical discord.
Equivalently,
$p_{AB}$ must be invariant under the measurement $M$ on $B$ side, which we formalize in the following proposition.
\begin{proposition}
\label{thm2}
The classical discord of the state $p_{AB}$ with noise $\mathcal{M}$ on $B$ is zero,
namely $\mathcal{D}^\mathcal{M}_{{A}\to{B}}(p_{AB}) = 0$, 
iff
\begin{equation}
\label{eq:zerodiscord}
	p_{AB'} = p_{AB}.
\end{equation}
\end{proposition}
We now investigate condition~\eqref{eq:zerodiscord} for zero classical discord in two cases:
i)~for a fixed noise $\mathcal{M}$, we find all possible states that have zero classical discord, and
ii)~for a given state $p_{AB}$ we show how to find all noisy channels that induce zero classical discord.

\emph{Case i)} We employ reshaping to find all possible states $p_{AB}$ that satisfy~\eqref{eq:zerodiscord}. A square matrix~$C$ can be reshaped into a column vector~$\bm{C}$
by employing a column-major order;
i.e., each column is stacked on top of its next adjacent column.
Arbitrary $A,B,C,D$ square matrices satisfy the identity
\begin{equation}
	ACB^T=D\iff (A\otimes B)\bm{C}=\bm{D}
\end{equation}
so Proposition~\ref{thm2} is equivalent to
\begin{equation}
\label{eq:Dreshaped}
	\mathcal{D}^\mathcal{M}_{{A}\to{B}}(p_{AB}) = 0 \iff (\mathds{1}\otimes M)\bm{p}_{AB}
		=\bm{p}_{AB}.
\end{equation}
In other words, $\bm{p}_{AB}$ must be a fixed point of $\mathds{1}\otimes M$
and must also be a valid probability vector.
Fixed points with such properties are called \emph{stationary vectors}.

Stochasticity of $M$ implies stochasticity of $\mathds{1}\otimes M$,
and the Perron-Frobenius theorem states that any stochastic matrix has at least one stationary vector. Therefore, a solution to the right-hand side~\eqref{eq:Dreshaped} always exists.
To find the class of all~$R$ possible states of zero discord,
let $\{\bm{m}_1,\bm{m}_2,\ldots,\bm{m}_R\}$
be the largest set of linearly independent stationary vectors of~$M$
($R=1$ if all entries of~$M$ are strictly positive).

Thus, a stationary vector of $\mathds{1}\otimes M$ must have the form
\begin{equation}
\label{eq:form}
	\sum_{j=0}^{2^n}\sum_{k=0}^R q_{jk}\mathds{1}_j\otimes\bm{m}_k,\;
	\sum_{j=0}^{2^n}\sum_{k=0}^R q_{jk}=1,\;
	q_{jk}\geq 0~\forall j,k,
\end{equation}
and $\mathds{1}_j$ denotes the $j^\text{th}$ eigenvector of $\mathds{1}$,
i.e., the vector comprising all zeros except for a unique~$1$ at the $j^\text{th}$ position.
Therefore, a class of zero-discord states $\mathcal{S}_0=\{\bm{p}_{AB}\}$ exists
with each~$\bm{p}_{AB}$ of the form~\eqref{eq:form}.

\emph{Case ii)} For a given fixed state $p_{AB}$, we now show how to find all possible 
noisy channels that induce zero classical discord.
The right-hand side of~\eqref{eq:Dreshaped} can be rewritten as
$(\mathds{1}\otimes M-\mathds{1})p_{AB}^T = 0$,
which holds iff 
\begin{equation}
\label{eq:channelclass}
	(\mathds{1}\otimes p_{AB})\bm{M}=\bm{p}^T_{AB}.
\end{equation}
Thus, given $p_{AB}$,
the class of channels for which the classical discord is zero can be found from solving~\eqref{eq:channelclass}
with the additional restriction that $\bm{M}$ is a valid representation of a classical channel
(i.e., its reshaping~$M$ is stochastic).
This last restriction is a linear constraint, and, therefore, the solutions to~\eqref{eq:channelclass} can be found using linear programming \cite{MatousekGartner:LinProg}.

Just as zero-discord quantum states have zero measure \cite{PhysRevA.81.052318}, zero-discord classical states 
have zero measure except in the singular case of perfect measurement, 
as formalized below. 
\begin{proposition}
The set of classical zero-discord states has measure zero in the set of all classical bipartite states
except for noiseless measurements.
\end{proposition}
\begin{proof}
A zero-discord state can be represented by the distribution~$\{q_{jk}\}$~\eqref{eq:form},
which has a $2^nR$-dimensional domain.
Unless $R=2^n$,
which pertains to $\bm{M}=\bm{\mathds{1}}$,
the domain of the distribution is strictly lower in dimension
so the set of zero-discord states has measure zero unless the channel is noiseless.
\end{proof}

We now have a definition of classical discord,
demonstrated its natural correspondence with quantum discord,
showed that nonzero classical discord is due to measurement noise on Bob's side,
and proved that states with zero classical discord occupy zero measure. 

While in the quantum case the natural stochasticity present on the ambiguous choice of the local measurement basis is sufficient to led to nonzero discord, in the classical case such stochasticity is not present, as explained above, and the discord is only different from zero if another kind of stochasticity, the noisy measurement, is present. The minimization present in~\eqref{disc} is absent in the classical discord definition due to this absence ambiguity in the measurement basis. However, if the theory is relaxed in a way that the stochastic matrix  $M$ is not fixed, a minimization over the set of allowed matrices $M$ is required.

\begin{remark} By combining~\eqref{eq:JM} and~\eqref{eq:DMABpAB} for unfixed channels ${\mathcal{M}}$ the classical stochastic discord reads
\begin{eqnarray}
\label{eq:discstoch}
\mathcal{D}_{A\rightarrow B}^{C}(p_{AB})&:=&\min_{\mathcal{M}} \mathcal{D}_{A\rightarrow B}^{\mathcal{M}}(p_{AB})\nonumber\\
&=&\min_{\mathcal{M}} H(A|B') -H(A|B),
\end{eqnarray}
which explicitly shows its usefulness in signaling how much the system state is disturbed by the (imperfect) measurement, similarly to the quantum discord. Throughout the rest of the text we deal only with $\mathcal{D}_{A\rightarrow B}^{\mathcal{M}}(p_{AB})$ to simplify the discussion.
\end{remark}

Now we establish an operational meaning via state merging, analogous to the operationalized
quantum discord.

\begin{definition}
State merging is a two-party task whereby the bipartite state~$p_{AB}$
comprising two shares of size~$n$ and~$m$ bits
is merged into a unipartite state of size~$n+m$ bits held by one party
such that the merged state has the same distribution as the original bipartite state.
\end{definition}
\begin{remark}
The term ``state merging'' arose in quantum information~\cite{HOW05,HOW07} but quantizes 
the notion of compressing correlated classical data streams~\cite{SW73a,SW73b}.
We use the term state merging for the classical case because our objective is to 
connect quantum discord with classical information theory so transferring quantum terminology
to the classical domain is appropriate here.
\end{remark}

We consider two-bit state merging 
(Alice holds one bit and Bob holds the other bit)
as this simple case suffices to establish operational
meaning for classical discord.
The following argument is readily extended to multiple bits held by each party.
In the quantum case,
the operational interpretation of quantum discord arises through a Gel'fand-Naimark purification
of the bipartite state~$\rho_{AB}$ with the resultant tripartite state~$\rho_{ABC}$ being pure~\cite{GN43}.

We purify $p_{AB}$ to a tripartite state according to Def.~\ref{def:pure}
except that in this tripartite case conditional purity means that $H(A|BC)=0$ and also for $B|CA$ and $C|AB$.
Specifically the purified state is
\begin{equation}\label{eq:classical_pure}
	p_{ABC}=(1-q)\chi(a,b,c)+q\chi(\bar{a},\bar{b},\bar{c}),\;
	a, b, c\in\{0,1\},
\end{equation}
with~$0\leqslant q\leqslant 1$ an arbitrary mixing parameter associated with noisy preparation
and~$\bar{a}$ denoting the logical negation of~$a$;
i.e., $\bar{0}=1$ and~$\bar{1}=0$.
As before,
Bob's measurement is through a noisy channel~$\mathcal M$.
We treat Charlie's measurement technology as equivalent to Bob's so
Charlie's measurement is also through~$\mathcal M$.

We now show that the quantum relation between discord and conditional entropy \cite{PhysRevA.84.012313}
holds in the classical case as well.
To show this equivalence,
we consider the tripartite $ABC$ system
and first determine the classical discord between~$A$ and~$C'$ (noisy~$C$)
and then the conditional entropy between~$A$ and $B'$.
Then we see that that the classical discord for~$p_{AC'}$
is the conditional entropy for~$p_{AB'}$.

By our definition of classical discord\eqref{eq:DMABpAB}
\begin{align}
\label{eq:DACHAB}
	\mathcal{D}^{\mathcal{M}}_{A\rightarrow C}(p_{AC})
		=& H(A) + H(C) - H(A, C) \notag\\
		& - \left[H(A) + H(C') - H(A, C')\right]\notag\\
		=& H(C) - H(A, C) + H(A,C')- H(C') \notag\\
		=& H(A|C') = H(A|B').
\end{align}
The last line of~\eqref{eq:DACHAB} follows from the previous line because a simple calculation yields $H(C) = H(A,C) =: h_2(q)$, the entropy of the single bit with flip probability $q$, and $H(A|C') = H(A|B')$ because of the symmetry of the state~\eqref{eq:classical_pure} and the fact that the noise on $B$'s side is assumed to be the same as the noise on $C$'s side.

\begin{remark}
In the quantum case, $\mathcal{D}^{\mathcal{M}}_{A\rightarrow C}(p_{AC})$~\eqref{eq:DACHAB}
is modified by including entanglement of formation
between A and B, and entanglement of formation must be zero classically as entanglement is forbidden.
Hence, in the quantum case, 
conditional entropy can be negative as shown in quantum state merging,
but not in the classical case due to entanglement of formation being zero.
\end{remark}
In the case of quantum state merging, the total entanglement that is consumed between $A$ and $B$ is equal to the discord between $A$ with measurements on $C$ \cite{HOW05,HOW07, PhysRevA.84.012313}
\begin{equation}\label{eqnv1}
D(A|C) = E_F(A:B) + S(A|B),
\end{equation}
where $E_F(\bullet)$ denotes the entanglement of formation between $A$ and $B$ and $S(A|B)$ is the conditional entropy of $A$ given $B$. The conditional entropy being negative implies leftover entanglement at the end of the merging protocol. In contrast, our relation \eqref{eq:DACHAB} has no term analogous to the entanglement of formation.

We have addressed the question of whether discord is exclusive to quantum systems or whether it can manifest also in a stochastic information theory. Our method for showing nonzero classical discord has been to construct
a model for stochastic information, which incorporates a noisy measurement
represented by a stochastic channel.
A direct consequence of measurement being noisy is the in-equivalence of defining
classical mutual information in terms of joint vs conditional entropy.
In an operational treatment of conditional measurement,
conditional entropy involves measuring a string and announcing these results to the other 
party.
Nonzero noise in the measurement causes mutual information based on conditional entropy
to differ from mutual information based on joint entropy.
Our findings clearly demonstrate that the notion of discord can be extended as a quantifier of stochasticity in general. In the nonquantum case stochasticity is introduced by imperfect or noisy measurements, while in the quantum case it is naturally introduced by the arbitrariness of the local measurement basis. In both situations the discord measures how much the joint state is affected by a local measurement.

A consequence of our classical discord investigation is that
the role of entanglement
is precisely on the negativity of conditional information.
This negativity only arises if entanglement of formation is nonzero.
Thus, discord could be understood classically through a stochastic information figure of merit
except when entanglement of formation is nonzero.
Operationally nonzero entanglement of formation is precisely what distinguishes our classical state merging from genuine quantum state merging~\cite{HOW05,HOW07}. 

The present discussion significantly contributes to reinforce the notion of discord as a measure of disturbance of a state under measurement. In addition to the fundamental issues that can be raised vis-\`a-vis the quantum discord and the quantum-classical separation, our definition of classical discord may be relevant for other well-established information-theoretical related areas such as measurement theory and statistical inference.

\begin{acknowledgments} 
	V.G. acknowledges financial support from PIMS, Industry Canada and NSERC. M.C.O. acknowledges support by FAPESP and CNPq through the National Institute for Science and Technology of Quantum Information (INCT-IQ) and the Research Center in Optics and Photonics (CePOF). 
	B.C.S. has been supported by AITF and NSERC.
	V.G. acknowledges useful discussions with M.~Piani, A.~Brodutch and J.~Combes.
\end{acknowledgments}


\begin{thebibliography}{23}%
\makeatletter
\providecommand \@ifxundefined [1]{%
 \@ifx{#1\undefined}
}%
\providecommand \@ifnum [1]{%
 \ifnum #1\expandafter \@firstoftwo
 \else \expandafter \@secondoftwo
 \fi
}%
\providecommand \@ifx [1]{%
 \ifx #1\expandafter \@firstoftwo
 \else \expandafter \@secondoftwo
 \fi
}%
\providecommand \natexlab [1]{#1}%
\providecommand \enquote  [1]{``#1''}%
\providecommand \bibnamefont  [1]{#1}%
\providecommand \bibfnamefont [1]{#1}%
\providecommand \citenamefont [1]{#1}%
\providecommand \href@noop [0]{\@secondoftwo}%
\providecommand \href [0]{\begingroup \@sanitize@url \@href}%
\providecommand \@href[1]{\@@startlink{#1}\@@href}%
\providecommand \@@href[1]{\endgroup#1\@@endlink}%
\providecommand \@sanitize@url [0]{\catcode `\\12\catcode `\$12\catcode
  `\&12\catcode `\#12\catcode `\^12\catcode `\_12\catcode `\%12\relax}%
\providecommand \@@startlink[1]{}%
\providecommand \@@endlink[0]{}%
\providecommand \url  [0]{\begingroup\@sanitize@url \@url }%
\providecommand \@url [1]{\endgroup\@href {#1}{\urlprefix }}%
\providecommand \urlprefix  [0]{URL }%
\providecommand \Eprint [0]{\href }%
\providecommand \doibase [0]{http://dx.doi.org/}%
\providecommand \selectlanguage [0]{\@gobble}%
\providecommand \bibinfo  [0]{\@secondoftwo}%
\providecommand \bibfield  [0]{\@secondoftwo}%
\providecommand \translation [1]{[#1]}%
\providecommand \BibitemOpen [0]{}%
\providecommand \bibitemStop [0]{}%
\providecommand \bibitemNoStop [0]{.\EOS\space}%
\providecommand \EOS [0]{\spacefactor3000\relax}%
\providecommand \BibitemShut  [1]{\csname bibitem#1\endcsname}%
\let\auto@bib@innerbib\@empty
\bibitem [{\citenamefont {Schr\"{o}dinger}(1935{\natexlab{a}})}]{Sch35a}%
  \BibitemOpen
  \bibfield  {author} {\bibinfo {author} {\bibfnamefont {Erwin}\ \bibnamefont
  {Schr\"{o}dinger}},\ }\bibfield  {title} {\enquote {\bibinfo {title} {Die
  gegenw\"{a}rtige situation in der quantenmechanik},}\ }\href@noop {}
  {\bibfield  {journal} {\bibinfo  {journal} {Naturwissenschaften}\ }\textbf
  {\bibinfo {volume} {23}},\ \bibinfo {pages} {807--812} (\bibinfo {year}
  {1935}{\natexlab{a}})}\BibitemShut {NoStop}%
\bibitem [{\citenamefont {Schr\"{o}dinger}(1935{\natexlab{b}})}]{Sch35b}%
  \BibitemOpen
  \bibfield  {author} {\bibinfo {author} {\bibfnamefont {Erwin}\ \bibnamefont
  {Schr\"{o}dinger}},\ }\bibfield  {title} {\enquote {\bibinfo {title} {Die
  gegenw\"{a}rtige situation in der quantenmechanik},}\ }\href@noop {}
  {\bibfield  {journal} {\bibinfo  {journal} {Naturwissenschaften}\ }\textbf
  {\bibinfo {volume} {23}},\ \bibinfo {pages} {823--828} (\bibinfo {year}
  {1935}{\natexlab{b}})}\BibitemShut {NoStop}%
\bibitem [{\citenamefont {Schr\"{o}dinger}(1935{\natexlab{c}})}]{Sch35c}%
  \BibitemOpen
  \bibfield  {author} {\bibinfo {author} {\bibfnamefont {Erwin}\ \bibnamefont
  {Schr\"{o}dinger}},\ }\bibfield  {title} {\enquote {\bibinfo {title} {Die
  gegenw\"{a}rtige situation in der quantenmechanik},}\ }\href@noop {}
  {\bibfield  {journal} {\bibinfo  {journal} {Naturwissenschaften}\ }\textbf
  {\bibinfo {volume} {23}},\ \bibinfo {pages} {844--849} (\bibinfo {year}
  {1935}{\natexlab{c}})}\BibitemShut {NoStop}%
\bibitem [{\citenamefont {Bennett}\ \emph {et~al.}(1993)\citenamefont
  {Bennett}, \citenamefont {Brassard}, \citenamefont {Cr\'epeau}, \citenamefont
  {Jozsa}, \citenamefont {Peres},\ and\ \citenamefont {Wootters}}]{BBC+93}%
  \BibitemOpen
  \bibfield  {author} {\bibinfo {author} {\bibfnamefont {Charles~H.}\
  \bibnamefont {Bennett}}, \bibinfo {author} {\bibfnamefont {Gilles}\
  \bibnamefont {Brassard}}, \bibinfo {author} {\bibfnamefont {Claude}\
  \bibnamefont {Cr\'epeau}}, \bibinfo {author} {\bibfnamefont {Richard}\
  \bibnamefont {Jozsa}}, \bibinfo {author} {\bibfnamefont {Asher}\ \bibnamefont
  {Peres}}, \ and\ \bibinfo {author} {\bibfnamefont {William~K.}\ \bibnamefont
  {Wootters}},\ }\bibfield  {title} {\enquote {\bibinfo {title} {Teleporting an
  unknown quantum state via dual classical and einstein-podolsky-rosen
  channels},}\ }\href {\doibase 10.1103/PhysRevLett.70.1895} {\bibfield
  {journal} {\bibinfo  {journal} {Phys. Rev. Lett.}\ }\textbf {\bibinfo
  {volume} {70}},\ \bibinfo {pages} {1895--1899} (\bibinfo {year}
  {1993})}\BibitemShut {NoStop}%
\bibitem [{\citenamefont {Modi}\ \emph {et~al.}(2012)\citenamefont {Modi},
  \citenamefont {Brodutch}, \citenamefont {Cable}, \citenamefont {Paterek},\
  and\ \citenamefont {Vedral}}]{RevModPhys.84.1655}%
  \BibitemOpen
  \bibfield  {author} {\bibinfo {author} {\bibfnamefont {Kavan}\ \bibnamefont
  {Modi}}, \bibinfo {author} {\bibfnamefont {Aharon}\ \bibnamefont {Brodutch}},
  \bibinfo {author} {\bibfnamefont {Hugo}\ \bibnamefont {Cable}}, \bibinfo
  {author} {\bibfnamefont {Tomasz}\ \bibnamefont {Paterek}}, \ and\ \bibinfo
  {author} {\bibfnamefont {Vlatko}\ \bibnamefont {Vedral}},\ }\bibfield
  {title} {\enquote {\bibinfo {title} {The classical-quantum boundary for
  correlations: Discord and related measures},}\ }\href {\doibase
  10.1103/RevModPhys.84.1655} {\bibfield  {journal} {\bibinfo  {journal} {Rev.
  Mod. Phys.}\ }\textbf {\bibinfo {volume} {84}},\ \bibinfo {pages}
  {1655--1707} (\bibinfo {year} {2012})}\BibitemShut {NoStop}%
\bibitem [{\citenamefont {Ferraro}\ \emph {et~al.}(2010)\citenamefont
  {Ferraro}, \citenamefont {Aolita}, \citenamefont {Cavalcanti}, \citenamefont
  {Cucchietti},\ and\ \citenamefont {Ac\'{i}n}}]{PhysRevA.81.052318}%
  \BibitemOpen
  \bibfield  {author} {\bibinfo {author} {\bibfnamefont {A.}~\bibnamefont
  {Ferraro}}, \bibinfo {author} {\bibfnamefont {L.}~\bibnamefont {Aolita}},
  \bibinfo {author} {\bibfnamefont {D.}~\bibnamefont {Cavalcanti}}, \bibinfo
  {author} {\bibfnamefont {F.~M.}\ \bibnamefont {Cucchietti}}, \ and\ \bibinfo
  {author} {\bibfnamefont {A.}~\bibnamefont {Ac\'{i}n}, \bibfnamefont {A.in}},\
  }\bibfield  {title} {\enquote {\bibinfo {title} {Almost all quantum states
  have nonclassical correlations},}\ }\href {\doibase
  10.1103/PhysRevA.81.052318} {\bibfield  {journal} {\bibinfo  {journal} {Phys.
  Rev. A}\ }\textbf {\bibinfo {volume} {81}},\ \bibinfo {pages} {052318}
  (\bibinfo {year} {2010})}\BibitemShut {NoStop}%
\bibitem [{\citenamefont {Merali}(2011)}]{Mer11}%
  \BibitemOpen
  \bibfield  {author} {\bibinfo {author} {\bibfnamefont {Zeeya}\ \bibnamefont
  {Merali}},\ }\bibfield  {title} {\enquote {\bibinfo {title} {Quantum
  computing: The power of discord},}\ }\href@noop {} {\bibfield  {journal}
  {\bibinfo  {journal} {Nature (London)}\ }\textbf {\bibinfo {volume} {474}},\
  \bibinfo {pages} {24--26} (\bibinfo {year} {2011})}\BibitemShut {NoStop}%
\bibitem [{\citenamefont {Bridgman}(1927)}]{Bri27}%
  \BibitemOpen
  \bibfield  {author} {\bibinfo {author} {\bibfnamefont {Percy~Williams}\
  \bibnamefont {Bridgman}},\ }\href@noop {} {\emph {\bibinfo {title} {The Logic
  of Modern Physics}}}\ (\bibinfo  {publisher} {Macmillan},\ \bibinfo {address}
  {New York},\ \bibinfo {year} {1927})\BibitemShut {NoStop}%
\bibitem [{\citenamefont {Horodecki}\ \emph {et~al.}(2005)\citenamefont
  {Horodecki}, \citenamefont {Oppenheim},\ and\ \citenamefont
  {Winter}}]{HOW05}%
  \BibitemOpen
  \bibfield  {author} {\bibinfo {author} {\bibfnamefont {Micha{\l}}\
  \bibnamefont {Horodecki}}, \bibinfo {author} {\bibfnamefont {Jonathan}\
  \bibnamefont {Oppenheim}}, \ and\ \bibinfo {author} {\bibfnamefont {Andreas}\
  \bibnamefont {Winter}},\ }\bibfield  {title} {\enquote {\bibinfo {title}
  {Partial quantum information},}\ }\href@noop {} {\bibfield  {journal}
  {\bibinfo  {journal} {Nature (London)}\ }\textbf {\bibinfo {volume} {436}},\
  \bibinfo {pages} {673--676} (\bibinfo {year} {2005})}\BibitemShut {NoStop}%
\bibitem [{\citenamefont {Horodecki}\ \emph {et~al.}(2007)\citenamefont
  {Horodecki}, \citenamefont {Oppenheim},\ and\ \citenamefont
  {Winter}}]{HOW07}%
  \BibitemOpen
  \bibfield  {author} {\bibinfo {author} {\bibfnamefont {Micha{\l}}\
  \bibnamefont {Horodecki}}, \bibinfo {author} {\bibfnamefont {Jonathan}\
  \bibnamefont {Oppenheim}}, \ and\ \bibinfo {author} {\bibfnamefont {Andreas}\
  \bibnamefont {Winter}},\ }\bibfield  {title} {\enquote {\bibinfo {title}
  {Quantum state merging and negative information},}\ }\href@noop {} {\bibfield
   {journal} {\bibinfo  {journal} {Commun. Math. Phys.}\ }\textbf {\bibinfo
  {volume} {269}},\ \bibinfo {pages} {107--136} (\bibinfo {year}
  {2007})}\BibitemShut {NoStop}%
\bibitem [{\citenamefont {Cavalcanti}\ \emph {et~al.}(2011)\citenamefont
  {Cavalcanti}, \citenamefont {Aolita}, \citenamefont {Boixo}, \citenamefont
  {Modi}, \citenamefont {Piani},\ and\ \citenamefont
  {Winter}}]{PhysRevA.83.032324}%
  \BibitemOpen
  \bibfield  {author} {\bibinfo {author} {\bibfnamefont {D.}~\bibnamefont
  {Cavalcanti}}, \bibinfo {author} {\bibfnamefont {L.}~\bibnamefont {Aolita}},
  \bibinfo {author} {\bibfnamefont {S.}~\bibnamefont {Boixo}}, \bibinfo
  {author} {\bibfnamefont {K.}~\bibnamefont {Modi}}, \bibinfo {author}
  {\bibfnamefont {M.}~\bibnamefont {Piani}}, \ and\ \bibinfo {author}
  {\bibfnamefont {A.}~\bibnamefont {Winter}},\ }\bibfield  {title} {\enquote
  {\bibinfo {title} {Operational interpretations of quantum discord},}\ }\href
  {\doibase 10.1103/PhysRevA.83.032324} {\bibfield  {journal} {\bibinfo
  {journal} {Phys. Rev. A}\ }\textbf {\bibinfo {volume} {83}},\ \bibinfo
  {pages} {032324} (\bibinfo {year} {2011})}\BibitemShut {NoStop}%
\bibitem [{\citenamefont {Madhok}\ and\ \citenamefont {Datta}(2011)}]{MS11}%
  \BibitemOpen
  \bibfield  {author} {\bibinfo {author} {\bibfnamefont {Vaibhav}\ \bibnamefont
  {Madhok}}\ and\ \bibinfo {author} {\bibfnamefont {Animesh}\ \bibnamefont
  {Datta}},\ }\bibfield  {title} {\enquote {\bibinfo {title} {Interpreting
  quantum discord through quantum state merging},}\ }\href {\doibase
  10.1103/PhysRevA.83.032323} {\bibfield  {journal} {\bibinfo  {journal} {Phys.
  Rev. A}\ }\textbf {\bibinfo {volume} {83}},\ \bibinfo {pages} {032323}
  (\bibinfo {year} {2011})}\BibitemShut {NoStop}%
\bibitem [{\citenamefont {Ferrie}\ and\ \citenamefont
  {Combes}(2014)}]{PhysRevLett.113.120404}%
  \BibitemOpen
  \bibfield  {author} {\bibinfo {author} {\bibfnamefont {Christopher}\
  \bibnamefont {Ferrie}}\ and\ \bibinfo {author} {\bibfnamefont {Joshua}\
  \bibnamefont {Combes}},\ }\bibfield  {title} {\enquote {\bibinfo {title} {How
  the result of a single coin toss can turn out to be 100 heads},}\ }\href
  {\doibase 10.1103/PhysRevLett.113.120404} {\bibfield  {journal} {\bibinfo
  {journal} {Phys. Rev. Lett.}\ }\textbf {\bibinfo {volume} {113}},\ \bibinfo
  {pages} {120404} (\bibinfo {year} {2014})}\BibitemShut {NoStop}%
\bibitem [{\citenamefont {Spekkens}(2007)}]{Spe07}%
  \BibitemOpen
  \bibfield  {author} {\bibinfo {author} {\bibfnamefont {Robert~W.}\
  \bibnamefont {Spekkens}},\ }\bibfield  {title} {\enquote {\bibinfo {title}
  {Evidence for the epistemic view of quantum states: A toy theory},}\ }\href
  {\doibase 10.1103/PhysRevA.75.032110} {\bibfield  {journal} {\bibinfo
  {journal} {Phys. Rev. A}\ }\textbf {\bibinfo {volume} {75}},\ \bibinfo
  {pages} {032110} (\bibinfo {year} {2007})}\BibitemShut {NoStop}%
\bibitem [{\citenamefont {Milonni}(1976)}]{Mil76}%
  \BibitemOpen
  \bibfield  {author} {\bibinfo {author} {\bibfnamefont {Peter~W.}\
  \bibnamefont {Milonni}},\ }\bibfield  {title} {\enquote {\bibinfo {title}
  {Semiclassical and quantum-electrodynamical approaches in nonrelativistic
  radiation theory},}\ }\href@noop {} {\bibfield  {journal} {\bibinfo
  {journal} {Phys. Rep.}\ }\textbf {\bibinfo {volume} {25}},\ \bibinfo {pages}
  {1--81} (\bibinfo {year} {1976})}\BibitemShut {NoStop}%
\bibitem [{\citenamefont {Cover}\ and\ \citenamefont
  {Thomas}(2005)}]{CoverThomas:ElementsOfInformationTheory}%
  \BibitemOpen
  \bibfield  {author} {\bibinfo {author} {\bibfnamefont {T.~M.}\ \bibnamefont
  {Cover}}\ and\ \bibinfo {author} {\bibfnamefont {J.~A.}\ \bibnamefont
  {Thomas}},\ }\href@noop {} {\emph {\bibinfo {title} {Elements of Information
  Theory}}},\ \bibinfo {edition} {2nd}\ ed.\ (\bibinfo  {publisher} {Wiley},\
  \bibinfo {address} {New York},\ \bibinfo {year} {2005})\BibitemShut {NoStop}%
\bibitem [{\citenamefont {Nielsen}\ and\ \citenamefont {Chuang}(2000)}]{NC00}%
  \BibitemOpen
  \bibfield  {author} {\bibinfo {author} {\bibfnamefont {Michael~A.}\
  \bibnamefont {Nielsen}}\ and\ \bibinfo {author} {\bibfnamefont {Isaac~L.}\
  \bibnamefont {Chuang}},\ }\href@noop {} {\emph {\bibinfo {title} {Quantum
  Computation and Quantum Information}}},\ \bibinfo {edition} {5th}\ ed.\
  (\bibinfo  {publisher} {Cambridge University Press},\ \bibinfo {address}
  {Cambridge},\ \bibinfo {year} {2000})\BibitemShut {NoStop}%
\bibitem [{\citenamefont {Ollivier}\ and\ \citenamefont {Zurek}(2001)}]{OZ01}%
  \BibitemOpen
  \bibfield  {author} {\bibinfo {author} {\bibfnamefont {Harold}\ \bibnamefont
  {Ollivier}}\ and\ \bibinfo {author} {\bibfnamefont {Wojciech~H.}\
  \bibnamefont {Zurek}},\ }\bibfield  {title} {\enquote {\bibinfo {title}
  {Quantum discord: A measure of the quantumness of correlations},}\ }\href
  {\doibase 10.1103/PhysRevLett.88.017901} {\bibfield  {journal} {\bibinfo
  {journal} {Phys. Rev. Lett.}\ }\textbf {\bibinfo {volume} {88}},\ \bibinfo
  {pages} {017901} (\bibinfo {year} {2001})}\BibitemShut {NoStop}%
\bibitem [{\citenamefont {Matou\v{s}ek}\ and\ \citenamefont
  {G\"artner}(2006)}]{MatousekGartner:LinProg}%
  \BibitemOpen
  \bibfield  {author} {\bibinfo {author} {\bibfnamefont {Ji\v{r}\'i}\
  \bibnamefont {Matou\v{s}ek}}\ and\ \bibinfo {author} {\bibfnamefont {Bernd}\
  \bibnamefont {G\"artner}},\ }\href@noop {} {\emph {\bibinfo {title}
  {Understanding and Using Linear Programming}}},\ Universitext\ (\bibinfo
  {publisher} {Springer-Verlag},\ \bibinfo {address} {Berlin},\ \bibinfo {year}
  {2006})\BibitemShut {NoStop}%
\bibitem [{\citenamefont {Slepian}\ and\ \citenamefont
  {Wolf}(1973{\natexlab{a}})}]{SW73a}%
  \BibitemOpen
  \bibfield  {author} {\bibinfo {author} {\bibfnamefont {D.}~\bibnamefont
  {Slepian}}\ and\ \bibinfo {author} {\bibfnamefont {J.~K.}\ \bibnamefont
  {Wolf}},\ }\bibfield  {title} {\enquote {\bibinfo {title} {Noiseless coding
  of correlated information sources},}\ }\href@noop {} {\bibfield  {journal}
  {\bibinfo  {journal} {IEEE Trans. Inf. Theory}\ }\textbf {\bibinfo {volume}
  {19}},\ \bibinfo {pages} {471--480} (\bibinfo {year}
  {1973}{\natexlab{a}})}\BibitemShut {NoStop}%
\bibitem [{\citenamefont {Slepian}\ and\ \citenamefont
  {Wolf}(1973{\natexlab{b}})}]{SW73b}%
  \BibitemOpen
  \bibfield  {author} {\bibinfo {author} {\bibfnamefont {D.}~\bibnamefont
  {Slepian}}\ and\ \bibinfo {author} {\bibfnamefont {J.~K.}\ \bibnamefont
  {Wolf}},\ }\bibfield  {title} {\enquote {\bibinfo {title} {A coding theorem
  for multiple access channels with correlated sources},}\ }\href@noop {}
  {\bibfield  {journal} {\bibinfo  {journal} {Bell Syst. Tech. J.}\ }\textbf
  {\bibinfo {volume} {52}},\ \bibinfo {pages} {1037--1076} (\bibinfo {year}
  {1973}{\natexlab{b}})}\BibitemShut {NoStop}%
\bibitem [{\citenamefont {Gelfand}\ and\ \citenamefont {Neumark}(1943)}]{GN43}%
  \BibitemOpen
  \bibfield  {author} {\bibinfo {author} {\bibfnamefont {Israel}\ \bibnamefont
  {Gelfand}}\ and\ \bibinfo {author} {\bibfnamefont {Mark}\ \bibnamefont
  {Neumark}},\ }\bibfield  {title} {\enquote {\bibinfo {title} {On the
  imbedding of normed rings into the ring of operators in hilbert space},}\
  }\href@noop {} {\bibfield  {journal} {\bibinfo  {journal} {Rec. Math. [Mat.
  Sbornik] N. S.}\ }\textbf {\bibinfo {volume} {12(54)}},\ \bibinfo {pages}
  {197--217} (\bibinfo {year} {1943})}\BibitemShut {NoStop}%
\bibitem [{\citenamefont {Fanchini}\ \emph {et~al.}(2011)\citenamefont
  {Fanchini}, \citenamefont {Cornelio}, \citenamefont {de~Oliveira},\ and\
  \citenamefont {Caldeira}}]{PhysRevA.84.012313}%
  \BibitemOpen
  \bibfield  {author} {\bibinfo {author} {\bibfnamefont {Felipe}\ \bibnamefont
  {Fanchini}}, \bibinfo {author} {\bibfnamefont {Marcio}\ \bibnamefont
  {Cornelio}}, \bibinfo {author} {\bibfnamefont {Marcos}\ \bibnamefont
  {de~Oliveira}}, \ and\ \bibinfo {author} {\bibfnamefont {Amir}\ \bibnamefont
  {Caldeira}},\ }\bibfield  {title} {\enquote {\bibinfo {title} {Conservation
  law for distributed entanglement of formation and quantum discord},}\ }\href
  {\doibase 10.1103/PhysRevA.84.012313} {\bibfield  {journal} {\bibinfo
  {journal} {Phys. Rev. A}\ }\textbf {\bibinfo {volume} {84}},\ \bibinfo
  {pages} {012313} (\bibinfo {year} {2011})}\BibitemShut {NoStop}%
\end{thebibliography}

%

\end{document}